\documentclass{IEEEtran4PSCC}
\usepackage{hyperref}
\ifCLASSINFOpdf
   \usepackage[pdftex]{graphicx}
\else
   \usepackage[dvips]{graphicx}
\fi

\usepackage{amsthm,amssymb}
\usepackage{cite,comment}
\usepackage{mathtools}
\usepackage{algorithm}
\usepackage{algpseudocode}
\usepackage{booktabs}
\usepackage{xcolor}
\usepackage{subcaption}
\usepackage{multirow}

\newcommand{\Ec}{\mathcal{E}}

\newcommand{\Gc}{\mathcal{G}}

\newcommand{\Nc}{\mathcal{N}}

\newcommand{\real}{\ensuremath{\mathbb{R}}}

\newcommand{\diag}{{\rm diag}}

\newtheorem{theorem}{Theorem}[section]

\newtheorem{lemma}[theorem]{Lemma}

\newtheorem{assumption}[theorem]{Assumption}
\theoremstyle{definition}
\newtheorem{remark}[theorem]{Remark}

\newcommand{\oprocendsymbol}{\hbox{$\bullet$}}
\newcommand{\oprocend}{\relax\ifmmode\else\unskip\hfill\fi\oprocendsymbol}

\let\originalleft\left
\let\originalright\right
\renewcommand{\left}{\mathopen{}\mathclose\bgroup\originalleft}
\renewcommand{\right}{\aftergroup\egroup\originalright}

\newcommand{\bP}[2][]{\Pr\ifthenelse{\isempty{#1}}{}{_{#1}}\left[#2\right]}
\newcommand{\bE}[2][]{\mathop\mathbb{E}\ifthenelse{\isempty{#1}}{}{_{#1}}\left[#2\right]}
\newcommand{\bI}[2][]{\mathop\mathbb{I}\ifthenelse{\isempty{#1}}{}{_{#1}}\left[#2\right]}
\newcommand{\Var}[2][]{ {Var}\ifthenelse{\isempty{#1}}{}{_{#1}}\left[#2\right]}

\makeatletter
\let\old@ps@headings\ps@headings
\let\old@ps@IEEEtitlepagestyle\ps@IEEEtitlepagestyle
\def\psccfooter#1{%
    \def\ps@headings{%
        \old@ps@headings%
        \def\@oddfoot{\strut\hfill#1\hfill\strut}%
        \def\@evenfoot{\strut\hfill#1\hfill\strut}%
    }%
    \def\ps@IEEEtitlepagestyle{%
        \old@ps@IEEEtitlepagestyle%
        \def\@oddfoot{\strut\hfill#1\hfill\strut}%
        \def\@evenfoot{\strut\hfill#1\hfill\strut}%
    }%
    \ps@headings%
}
\makeatother

\psccfooter{%
        \parbox{\textwidth}{\hrulefill \\ \small{24th Power Systems Computation Conference} \hfill \begin{minipage}{0.2\textwidth}\centering \vspace*{4pt} \includegraphics[scale=0.06]{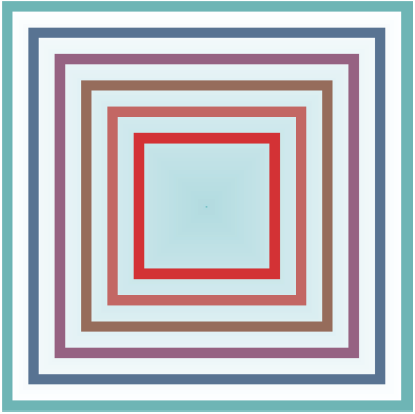}\\\small{PSCC 2026} \end{minipage} \hfill \small{Limassol, Cyprus --- June 8-12, 2026}}%
}

\begin{document}

\title{Efficient Policy Adaptation for Voltage Control Under Unknown Topology Changes}

\author{
\IEEEauthorblockN{Jie Feng}
\IEEEauthorblockA{
University of California San Diego\\
jif005@ucsd.edu}
\and
\IEEEauthorblockN{ Yuanyuan Shi}
\IEEEauthorblockA{
University of California San Diego\\
yyshi@ucsd.edu}
\and
\IEEEauthorblockN{Deepjyoti Deka}
\IEEEauthorblockA{
Massachusetts Institute of Technology\\
deepj87@mit.edu}
}

\maketitle

\begin{abstract}
Reinforcement learning (RL) has shown great potential for designing voltage control policies, but their performance often degrades under changing system conditions such as topology reconfigurations and load variations. 
We introduce a topology-aware online policy optimization framework that leverages data-driven estimation of voltage–reactive power sensitivities to achieve efficient policy adaptation. Exploiting the sparsity of topology-switching events, where only a few lines change at a time, our method efficiently detects topology changes and identifies the affected lines and parameters, enabling fast and accurate sensitivity updates without recomputing the full sensitivity matrix. The estimated sensitivity is subsequently used for online policy optimization of a pre-trained neural-network-based RL controller. Simulations on both the IEEE 13-bus and SCE 56-bus systems demonstrate over 90\% line identification accuracy, using only 15 data points. The proposed method also significantly improves voltage regulation performance compared with non-adaptive policies and adaptive policies that rely on regression-based online optimization methods for sensitivity estimation.
\end{abstract}


\section{Introduction}
Voltage regulation is a fundamental requirement for reliable distribution grid operations, but it has become increasingly challenging with the rapid growth of distributed energy resources (DERs) \cite{2017volt-survery}. The variability of renewable generation and load, together with the fast dynamics introduced by inverter-based resources, creates frequent and unpredictable voltage fluctuations that traditional devices such as on-load tap changers and capacitor banks cannot address effectively. In contrast, smart inverters are naturally suited for real-time voltage control, which can adjust reactive or real power flexibly based on real-time voltage measurements. 

A wide range of inverter-based voltage control strategies has been explored in the literature, broadly categorized as centralized, decentralized, and distributed \cite{2017volt-survery,PS-RH-VJN-VV-AMA-AKS:23}. Centralized approaches, often formulated as optimal power flow (OPF) problems \cite{Farivar-2012-VVC-PES,chen_input_2020} or model predictive control (MPC) schemes \cite{guo2019-mpc,maharjan-2021-mpc}, can achieve optimal coordination of devices but are computationally intensive. Moreover, these approaches require accurate system models and reliable communication infrastructure, which are often unavailable in practical distribution grids. At the other extreme, decentralized controllers, such as standardized droop-based rules \cite{IEEE1547_2018,zhu2016fast,8003321}, provide fast, local responses at scale but may yield suboptimal setpoints and struggle to adapt to time-varying topology changes. Distributed optimization methods, including consensus-based updates \cite{BZ-AYSL-ADDG-DT:14,liu-distriuted-2018} and primal–dual algorithms \cite{7967778}, seek to balance coordination and scalability by leveraging limited communication among neighboring devices. but require careful parameter tuning for iterative convergence. Despite their differences, the effectiveness and stability of these approaches ultimately depend on accurate and up-to-date network models.

When accurate network models are unavailable, learning-based controllers provide a compelling alternative. In particular, reinforcement learning (RL) enables controllers to derive effective control policies directly through interactions with the environment, eliminating the need for explicit system modeling. Once trained, RL controllers can capture nonlinear dynamics and operate in real time \cite{9721402}. Recent studies~\cite{10336939,11164977,JF-WC-JC-YS:23-csl} have further incorporated stability guarantees into RL-based voltage control by enforcing monotonicity constraints on neural policies to ensure Lyapunov stability. However, existing RL-based voltage control methods lack the ability to adapt to time-varying network topologies, as they are typically trained offline and deployed online without policy adaptation.

Complementing RL-based approaches, data-driven voltage control methods can leverage online measurements to adapt to evolving system conditions. These methods rely on voltage–control sensitivity information, which characterizes how voltage magnitudes respond to control actions. The sensitivities are typically estimated from data and subsequently used to update controllers or directly compute control inputs~\cite{8873667}. Most existing approaches estimate sensitivities using ordinary least squares~\cite{8383953}, ensemble regression~\cite{8412143}, or recursive least squares~\cite{6687276,picallo2022adaptive,7997787}. We refer readers to the recent review in~\cite{10934057} for a comprehensive overview. Similarly, topology estimation methods~\cite{7875102,park2018exact,8383953,park2020learning,deka2023learning} aim to recover network connectivity and line parameters from data, which can also provide sensitivity information. However, these methods typically assume a fixed network topology during operation and are not designed for real-time topology change detection.

In practice, distribution networks are frequently reconfigured through switching operations for reliability, maintenance, or fault isolation~\cite{deka2023learning}. To address the problem of time-varying system models, several topology change detection methods have been proposed, such as trend-vector analysis using µPMU data~\cite{7286490}, graph-signal-processing-based detection~\cite{9640113}, residual- and $\chi^2$-based tests integrated with state estimation~\cite{4543032}, and covariance analysis~\cite{deka2020graphical,10510511}. While these methods have improved detection accuracy, they often depend on full observability using Phasor Measurement Units, require knowledge of possible reconfiguration scenarios, or cannot distinguish load variations from true topology changes, limiting their applicability within closed-loop voltage control.

\begin{figure}
    \centering
    \includegraphics[width=\linewidth]{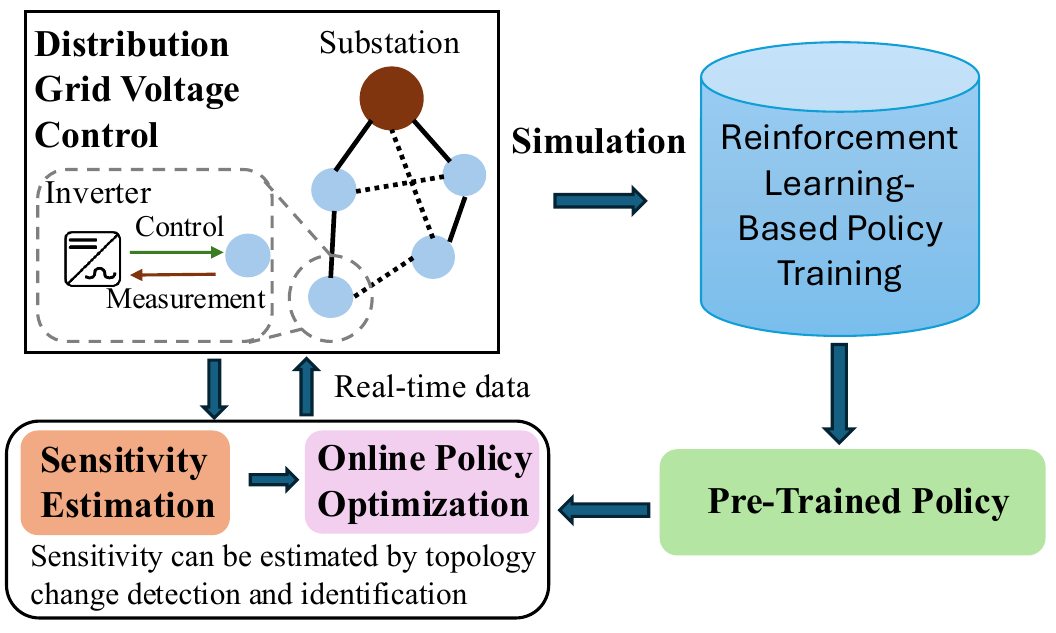}
    \caption{Overview of the proposed framework.} 
    \label{fig:overview}
\end{figure}

\subsection{Contributions}
This work introduces an efficient online policy adaptation framework for voltage control under unknown topology changes. 
The proposed method leverages the radial structure of distribution networks and the sparsity of topology changes to enable fast topology detection and sensitivity estimation using only a few voltage measurements and control actions. The estimated sensitivities are then incorporated into an online policy optimization framework~\cite{lin2024online}, which adapts a pretrained neural network–based RL controller to the new topology and minimizes real-time control costs. Figure~\ref{fig:overview} provides an overview of the framework. Our main contributions include: 
\begin{enumerate}
\item We propose an efficient algorithm for topology change detection and identification in radial distribution grids, capable of accurately localizing changed lines, reconstructing topology, and estimating parameters within a short observation window.
\item We integrate this algorithm with an online policy optimization framework, enabling real-time adaptation of neural network–based voltage controllers to load variations and topology reconfigurations.
\item We validate the framework in nonlinear simulations on the IEEE 13-bus and SCE 56-bus systems, achieving over 90\% accuracy in line-change identification, a 25\% reduction in control cost for the SCE 56-bus system, and 75\% reduction in sensitivity matrix estimation error.
\end{enumerate}

The remainder of the paper is organized as follows. 
Section~\ref{sec:pre} introduces the preliminaries, including the distribution grid model, voltage control formulation, and existing sensitivity estimation methods. Section~\ref{sec:topo_detect} presents the proposed topology change detection and identification algorithm, followed by the online policy optimization framework in Section~\ref{sec:online voltage control}. Section~\ref{sec:experiments} reports the numerical results, and conclusions are drawn in Section~\ref{sec:conclusion}.

\section{Preliminary}\label{sec:pre}
In this section, we first introduce the distribution grid (DG) model and the voltage control problem formulation, followed by a brief overview of sensitivity estimation necessary for data-driven voltage control.
\subsection{Distribution Grid Modeling}
A radial single-phase (or a balanced three-phase) DG having $N+1$ buses can be modeled by a tree graph $\Gc=(\Nc,\Ec)$ rooted at the substation. The set $\Nc_0:=\{0,\ldots,N\}$ contains all buses, and $\Ec$ denotes the set of lines. The substation node, labeled as 0, is modeled as an ideal voltage source fixed at 1 p.u.. We 
let $\Nc = \Nc_0 \setminus \{0\}$. For each bus $i\in \Nc$, let $v_i\in \real$ be its voltage magnitude, and $p_i,q_i\in \real$ its real and reactive power injections, respectively. 
We collect these quantities into vectors $ v,p, q \in \real^N$ for buses $1,2,\dots N$.
For each line $(m,n)$, let $r_{mn}$ and $x_{mn}$ denote its resistance and reactance, and $P_{mn}$ and $Q_{mn}$ the real and reactive power from bus $m$ to $n$.
Using the linearized DistFlow (LinDistFlow) approximation~\cite{HZ-HJL:15}, the power flow model in radial distribution networks can be represented as 
\begin{subequations}\label{eq:linear_pf}
    \begin{align}
&P_{mn}-\smashoperator{\sum_{(n,k)\in\Ec}} P_{n k}  =-p_n, ~~~
Q_{mn}-\smashoperator{\sum_{(n,k) \in \Ec}} Q_{n k}  =-q_n, \\
&v_m-v_n =2(r_{mn} P_{mn}+x_{mn} Q_{mn}),
\end{align}
\end{subequations}
which can be written in the following compact vector form
 \begin{align}\label{eq:v=Rp+Xq}
 v = R p + X q + v_0 {1}=X q +v_{env}.
 \end{align}

Here, the system matrices $R = [R_{mn}]_{N \times N}$ and $X = [X_{mn}]_{N \times N}$ with $R_{mn} := 2 \sum_{(i,j) \in \Ec_m \cap \Ec_n} r_{ij}$ and $X_{mn} := 2 \sum_{(i,j) \in \Ec_m \cap \Ec_n} x_{ij}$, where $\Ec_h \subseteq \Ec$ contains the lines on the unique path from bus $0$ to bus $h$. The voltage can be separated into the controllable part by $X q$ and the uncontrollable part $v_{env} = R p +v_0 {1}$. Particularly, we have that $X,R \succ 0$. 
Note that we only use the linear model \eqref{eq:linear_pf} for theoretical analysis, the experiments are conducted with the original nonlinear DistFlow model \cite{MEB-FFW:89}.

\subsection{Voltage Control Problem Formulation}
The goal of voltage control is to restore the voltage deviation with minimal control efforts. Given the voltage measurement $v_t$, the controller incrementally adjusts reactive power injection $q_{t+1}$, leading to a new voltage profile $v_{t+1}$. We focus on a controllable subset of $M\leq N$ buses ${\mathcal{P}}\in\mathcal{N}$ with reactive power $q_{\mathcal{P},t}$, where each controllable node regulates its reactive power based solely on local voltage measurements. This decentralized structure ensures reliable and scalable operation even without continuous communication.

The voltage control problem is formulated as follows:
\begin{subequations}\label{eq:volt-control-problem}
    \begin{align}
    \min_{\theta}& \quad  J(\theta)= \sum_{t=0}^{T}h_{t+1}, \label{eq:question,cost}\\
    \text{s.t.}\; &h_{t+1} = v_{\mathrm{dev},t+1}^\top  {Q}_x v_{\mathrm{dev},t+1} 
     + {q}_{\mathcal{P}, t+1}^\top  {Q}_u  {q}_{\mathcal{P}, t+1}, \hfill \label{eq:cost}\\
&{v}_{t+1} = {X}_{\mathcal{P}} {q}_{\mathcal{P},t+1} + Xq_{\mathrm{nc}} + {v}_{\mathrm{env}}\label{equ:dynamics:power_flow}\\
& q_{\mathcal{P},t+1}=q_{\mathcal{P},t}+u_t, \hfill \label{eq:q_update}\\
& {u}_t =\pi( {v}_t,\theta_t),  \hfill \label{equ:dynamics-simple:policy}
    \end{align}
\end{subequations}
where $v_{\mathrm{dev},t} := v_t - v_{\mathrm{nom}}$ is the voltage deviation, and $v_{\text{nom}}=1$ is the nominal voltage. ${Q}_x\in \mathbb{R}^{N \times N}$, ${Q}_u\in \mathbb{R}^{M \times M}$ are positive definite cost matrices, which define the voltage deviation and reactive injection costs. 
  
Equation \eqref{equ:dynamics:power_flow} represents the power flow model. $q_{nc}\in \mathbb{R}^{N}$ denotes the uncontrolled reactive power injections, and ${X}_{\mathcal{P}}\in\mathbb{R}^{N \times M}$ is the sensitivity of voltage magnitude with respect to the controllable reactive power injections, i.e., ${X}_{\mathcal{P}}$ is a sub-matrix obtained by keeping only the columns of $X$ indexed by $\mathcal{P}$.
Equation \eqref{eq:q_update} defines the reactive power update law, where the next reactive power setpoint $q_{\mathcal{P},t+1} \in \mathbb{R}^M$ is determined by the previous setpoint $q_{\mathcal{P},t}$ and the control action $u_t\in\mathbb{R}^M$. 
The control action is given by \eqref{equ:dynamics-simple:policy}, where $\pi( {v}_t,\theta_t)$ is a decentralized monotone neural-network controller parameterized by $\theta_t$, as in \cite{10336939}. 
Note that the time index on $\theta_t$ indicates that the controller parameters can be updated in real time.  Figure~\ref{fig:causal} depicts the causal relationships among system states, control inputs, policy parameters, and costs under the proposed algorithm.
The detailed procedure for online policy parameter updates is presented in Section \ref{sec:online voltage control}.

\begin{figure}
    \centering
    \includegraphics[width=\linewidth]{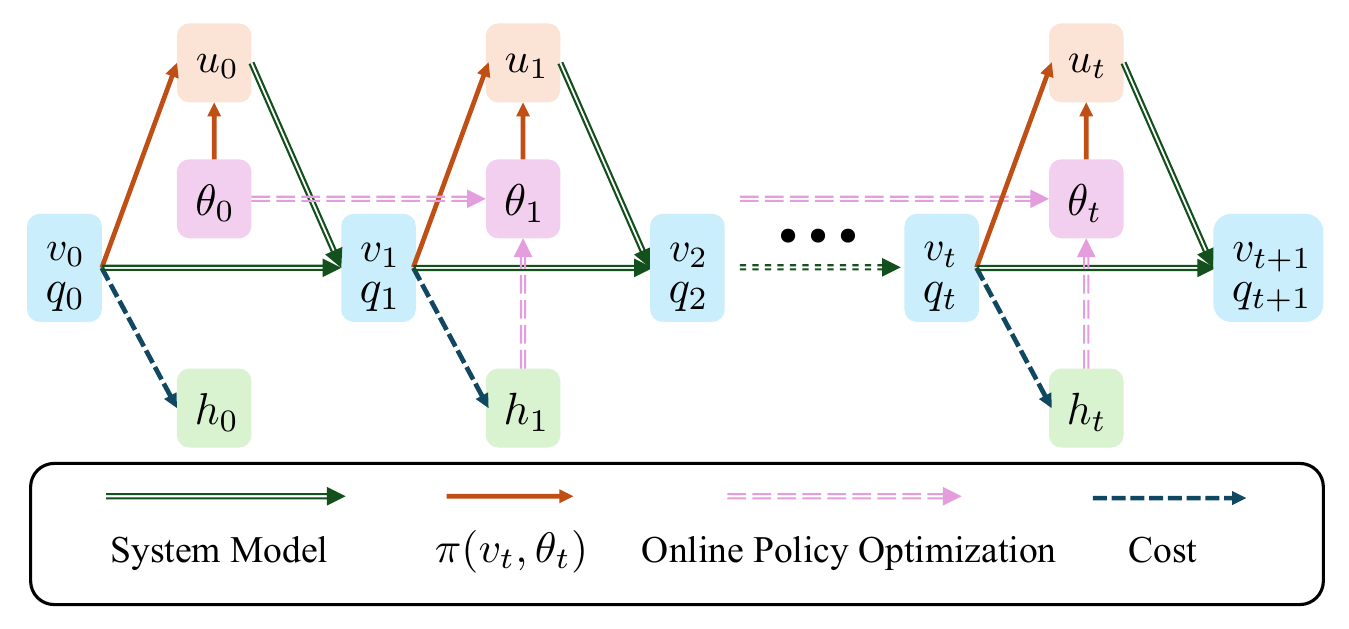}
    \caption{Diagram of the causal relationships between states, policy parameters, control inputs, and costs when online policy optimization is active. The policy parameters are updated after observing $h_t$ and applied to compute $u_t$.}
    \label{fig:causal}
\end{figure}

Notably, the voltage-control sensitivity matrix $X_{\mathcal{P}}$ in \eqref{equ:dynamics:power_flow} is critical for data-driven voltage control \cite{9302942,8873667,picallo2022adaptive,8412143, 8003321}, including ours, as it captures how reactive-power adjustments at the controllable buses affect the voltage magnitudes~\cite{abdelkader2024advancements}. 
The sensitivity matrix explicitly depends on the network topology. In the next section, we present existing methods for sensitivity estimation that can be used for data-driven voltage control under topology changes. 

\subsection{Sensitivity Estimation for Data-Driven Voltage Control}\label{sec:sub_sensitivity}
Define the voltage change between consecutive time steps as $\tilde{v}_t = v_{t+1}-v_t$. From \eqref{equ:dynamics:power_flow} and \eqref{eq:q_update}, we obtain
\begin{equation}\label{eq:vlt_diff}
    \tilde{v}_{t+1} = X_{\mathcal{P}} u_t.
\end{equation}
Given a set of historical data  $\{(u_t, \tilde{v}_{t+1})\}_{t=0}^{T}$, 
the sensitivity matrix $X_{\mathcal{P}}$ can be estimated by solving the following least-squares regression problem:
\begin{equation}
    \min_{X_{\mathcal{P}}}\; J(X_{\mathcal{P}}) = 
    \sum_{t=0}^{T}\| \tilde{v}_{t+1} - X_{\mathcal{P}} u_t \|_2^2.
\end{equation}

Several estimation strategies have been proposed in the literature for identifying the sensitivity matrix. The \emph{ordinary least squares} (OLS) method provides a batch estimate of $X_{\mathcal{P}}$ using historical voltage–control data under the assumption of a fixed network topology \cite{8383953, 8873667}. However, the presence of topology changes during data collection and the inherent collinearity in power system operational data \cite{7997787} can lead to degradation in its estimation performance.
The \emph{recursive least squares} (RLS) method improves upon OLS by enabling online updates as new data arrive \cite{picallo2022adaptive}, yet it remains limited in handling abrupt topology reconfigurations. For completeness, the details of the OLS and RLS methods are presented in Appendix \ref{appendix:sensitivity}. 
To address the aforementioned challenges, in this work, we propose a \emph{topology-aware estimation} approach that exploits the distribution network structure to detect and identify topology changes in real time, update the sensitivity matrix estimation, and use it for online policy optimization. 
\begin{remark}
    For clarity of derivation, Sections~\ref{sec:topo_detect} and~\ref{sec:online voltage control} adopt the LinDistFlow model with $\mathcal{P}=\mathcal{N}$, i.e., all buses are assumed to be controllable. In simulations, however, only a subset of buses has control capability. When controllability is partial, the topology change detection in Section~\ref{sec:topo_detect} remains valid because it relies solely on voltage-deviation estimation and does not require direct actuation at every bus. As a result, global topology changes can still be identified. Given the detected change and the original topology, the updated topology can be reconstructed. For online policy optimization, the reduced sensitivity matrix $X_{\mathcal{P}}$ is then derived from the estimated topology and used for control.
\end{remark}

\section{Topology Change Detection And Identification}
\label{sec:topo_detect}
This section introduces the proposed topology change detection and identification algorithm, which enables real-time detection of reconfiguration events and efficient identification of the affected lines and their parameters. 
For simulations, the algorithm is tested under the nonlinear DistFlow model, where only a subset of nodes is controlled, as detailed in Section~\ref{sec:experiments}.

We begin by outlining the assumptions under which the algorithm operates.
\begin{assumption}\label{asm:separation}
    There exists a timescale separation such that real-time voltage control dynamics evolve much faster than the timescales of load variations and topology changes.
\end{assumption}
This timescale separation assumption is standard \cite{2017volt-survery} as measurements and control actions can be executed on sub-second timescales, while load variations typically evolve more slowly and topology reconfigurations occur even less frequently. This separation ensures that topology changes can be reliably detected and addressed using measurements of voltage control. 

\begin{assumption}\label{asm:radial}
    The distribution grid operates in a radial configuration both before and after each topology reconfiguration event.
\end{assumption}
To ease protection and control, most distribution grids are operated in a radial (tree-like) topology \cite{kersting2018distribution}. In this work, we focus on the topology change detection for such systems.

\begin{assumption}\label{asm:known}
    The initial network topology and line parameters (denoted $X_0$) are assumed to be known. However, the timing and details of all topology reconfigurations, including the lines affected and the parameters of any newly added lines, are unknown.
\end{assumption}
This assumption is consistent with practice: system operators typically possess knowledge of a nominal system topology and parameters or can infer them through routine topology estimation. However, unforeseen reconfigurations resulting from faults or unrecorded topology switching actions remain uncertain.

\subsection{Topology Change Detection}
\label{subsec:topo detect}
The first step is to detect the occurrence of a topology change event. Since both load variations and topology changes can lead to voltage deviations, we discuss methods to differentiate between them.

Following \eqref{eq:v=Rp+Xq}, the pre-event voltage at time $t-1$ is
$$v_{t-1} = R_{t-1}p_{t-1} + X_{t-1}q_{t-1} + v_0 {1}.$$
By Assumption~\ref{asm:known}, $X_{t-1}$ is known. Suppose that at time $t$, a load change (real-power variation) or a topology change occurs, yielding
$$v_t = R_t p_t + X_t q_t + v_0 {1},$$
where $R_t = R_{t-1} + \Delta R$, $X_t = X_{t-1} + \Delta X$, denote the possible change in $R,X$ due to topology change. $p_t = p_{t-1} + \Delta p$ represents the possible load variation, and $q_t = q_{t-1} +  u _{t-1}$, where $u_t$ is the control action defined in Equation \eqref{equ:dynamics-simple:policy}\footnote{For simplicity, reactive power variations in non-controllable loads ($q_{nc}$ in Eq \eqref{equ:dynamics:power_flow}) are omitted, as they can be handled like real power changes.}. The corresponding voltage difference is
$$\tilde{v}_{t} = v_t - v_{t-1} 
= R_t\Delta p + \Delta R p_{t-1} + X_t u _{t-1} + \Delta X q_{t-1}.$$
If no further topology or load changes occur at $t+1$, then
\begin{equation}\label{eq:voltage deviation}
    \tilde{v}_{t+1} = v_{t+1} - v_t = X_t u _t.
\end{equation}

Assuming the controller is unaware of the change between time steps $t$ and $t+1$, it continues to predict the voltage changes using the previously given topology matrix $X_{t-1}$: 
$$\tilde{v}_{t}^{\mathrm{pred}} = X_{t-1} u _{t-1}, 
\qquad
\tilde{v}_{t+1}^{\mathrm{pred}} = X_{t-1} u _t.$$
We define the prediction error as $e_t = \tilde{v}_t - \tilde{v}_t^{\text{pred}}$.  
The behavior of this error can be used both to determine whether a change has occurred and to distinguish between a topology reconfiguration and a load variation, as formalized below.

\begin{lemma}[Detection Criterion under LinDistFlow Model]
\label{def:detection}
A change event is detected at time $t$ if $\|e_{t}\|_2 >0$. If $\|e_{t+1}\|_2 >0$, the change is due to a topology change, otherwise if $\|e_{t+1}\|_2=0$, the change is attributed to a load change. 
\end{lemma}
\textbf{Proof:} Under Assumption \ref{asm:separation}, when topology change occurs, $\Delta X, \Delta R\neq 0$. Under load variation without topology change $\Delta X, \Delta R =0$ and $\Delta p\neq 0$. Using the expression for $\tilde{v}_{t}, \tilde{v}_{t}^{\mathrm{pred}}$, it is clear that $e_{t}\neq 0$ when a topology change or load change occurs at time $t$ and hence $\|e_{t}\|_2 >0$. If it was due to a topology change, then $e_{t+1}=\tilde{v}_{t+1}-\tilde{v}_{t+1}^{\mathrm{pred}} = \Delta Xu_t \neq 0$. However if it was due to a load change, then $e_{t+1}= \tilde{v}_{t+1}-\tilde{v}_{t+1}^{\mathrm{pred}} = 0$. Thus, the result holds.
  
\textbf{Extension to nonlinear power flow model:} Nonlinearities in the power flow equations cause prediction errors, which can affect detection accuracy. To mitigate this, we propose a practical approach based on robust baselines~\cite{hampel1974influence} as below.

\begin{enumerate}
\item At each time $t$, compute the error norm $\lVert e_t\rVert_2$, and store it in a history buffer $B_t = \{\|e_{t-H}\|_2, ..., \|e_{t}\|_2\}$, where $H$ is the length of the error buffer.
\item Compute the median and median absolute deviation (MAD), i.e., $\text{med}=\text{median}(B_t),\,\text{MAD}=\text{median}(|B_t-\text{med}|)$. The robust baseline is defined as $\kappa_t = \text{med}+3.5\times\text{MAD}$. This median–MAD approach is a standard technique for robust anomaly detection, where the ratio typically ranges between 2 and 4.
\item Flag a change event when $\lVert e_t\rVert_2>\kappa_t$. Record the detection time, and proceed to the next step to distinguish whether it is a topology change or a load change event.
\item At time $t+1$, if the error remains, i.e., $\lVert e_{t+1}\rVert_2>\kappa_t$, we label it as a topology change; otherwise, we label it as a load change.
\end{enumerate}
The benefit of this mechanism is that it uses a robust baseline derived from the recent buffer and thus enables reliable real-time event detection without manual threshold tuning. 

\subsection{Topology Change Identification}
\label{sec:identification}
Once a topology change is detected, we identify the specific line changes and their parameters. An illustrative diagram of topology change is presented in Figure \ref{fig:topo_change_diag}.

\begin{figure}[tb]
    \centering
    \includegraphics[width=0.75\linewidth]{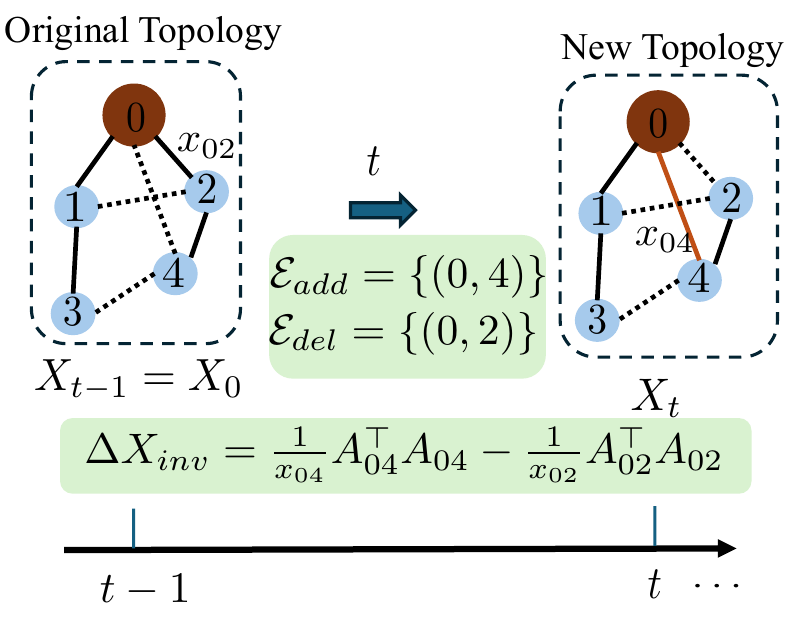}
    \caption{Diagram of a topology change event at time $t$. }
    \label{fig:topo_change_diag}
\end{figure}

\subsubsection{Residual-based topology change modeling}
Suppose a topology change occurs at time $t$. The voltage deviation at time $t+1$ is given by \eqref{eq:voltage deviation}. 
In a radial distribution grid, the sensitivity matrix $X_t$ is invertible \cite{deka2023learning}. Assuming no additional topology change after time $t$ as per Assumption \ref{asm:separation}, we have $X_{t+1} = X_t$ and the control input
$u _{t}=X_{t}^{-1} \tilde{v}_{t+1}$
with
\begin{equation}
\label{eq:X_decomp}
X_{t}^{-1} = A^\top D_x^{-1}A=\sum_{(i,j)\in\mathcal{E}}x_{ij}^{-1}A_{ij}^\top A_{ij}~,
\end{equation}
where $A$ is the reduced incidence matrix, $D_x=\diag({x_{ij}})$, and $A_{ij}=e_i^\top-e_j^\top$ for line $(i,j)$. Since $A_{ij}$ is supported only on buses $i$ and $j$, each term $A_{ij}^\top A_{ij}$ introduces nonzero entries only between directly connected buses. Consequently, $X_{t}^{-1}$ is sparse, with its sparsity pattern matching the grid topology.

Define the difference in the inverse sensitivity matrices as $\Delta X_{inv}=X_{t}^{-1}-X_{t-1}^{-1}$, where $X_{t-1}=X_{0}$ corresponds to the initial network topology. Following the decomposition in \eqref{eq:X_decomp},
\begin{equation}
\label{eq:X_inv}
    \Delta X_{inv} \!=\! \sum_{(i,j)\in\mathcal{E}_{\mathrm{add}}}\frac{1}{x_{ij}}A_{ij}^\top A_{ij}\!-\!\sum_{(k,m)\in\mathcal{E}_{\mathrm{del}}}\frac{1}{x_{km}}A_{km}^\top A_{km},
\end{equation}
where $\mathcal{E}_{\mathrm{add}}$ and $\mathcal{E}_{\mathrm{del}}$ denote the sets of added and deleted lines during the topology change event, respectively. To preserve the radial network structure, the number of added and deleted lines should be the same, i.e.,  $|\mathcal{E}_{\mathrm{add}}|=|\mathcal{E}_{\mathrm{del}}|$.

Given control $u_{t-1}$, the pre-event model--based on outdated topology information--predicts the voltage deviation as $\tilde v_{t}^{\mathrm{pred}}=X_0 u_{t-1}$, while the true voltage deviation under the new topology is $\tilde v_{t}=X_t u_{t-1}$. 
Since the predicted and the true voltage deviations differ due to the topology change, we can relate them through the control input, 
\begin{equation}\label{eq:u-vpre}
u_{t-1}=X_0^{-1}\tilde v_{t}^{\mathrm{pred}}=X_t^{-1}\tilde v_{t}=(X_0^{-1}+\Delta X_{\mathrm{inv}})\tilde v_{t}
\end{equation}
Equation \eqref{eq:u-vpre} links the predicted and true voltage deviations, which makes it possible to isolate the effect of the topology change. We capture it through the following residual,
\begin{equation}
\label{eq:residual_0}
    r_{t}^{res}=X_0^{-1}(\tilde{v}_{t}^{pred}-\tilde{v}_{t})=\Delta X_{inv}\tilde{v}_{t}.
\end{equation}
To make the structure of $\Delta X_{inv}$ defined in \eqref{eq:X_inv} explicit, we collect the incidence vectors of the candidate added and deleted lines into $E = [A_{ij}^\top, A_{km}^\top]\in\mathbb{R}^{N\times (|\mathcal{E}_{\mathrm{add}}|+|\mathcal{E}_{\mathrm{del}}|)}$, and define $\Gamma=\diag(\{x_{ij}^{-1}\}_{\forall (i,j)\in\mathcal{E}_{\mathrm{add}}},\{-x_{km}^{-1}\}_{\forall(k,m)\in\mathcal{E}_{\mathrm{del}}})$.
Substituting these definitions into \eqref{eq:residual_0}, the \emph{residual} term can be re-written as 
\begin{equation}\label{eq:residual}
    r_{t}^{\mathrm{res}}=E\Gamma E^\top \tilde{v}_{t}.
\end{equation}
The residual vector $r_{t}^{\mathrm{res}} \in \mathbb{R}^N$ lies in the range space of $E$ and is \emph{sparse} -- only nodes that are connected to the changed lines exhibit non-zero residuals. Consequently, the residual $r_{t}^{\mathrm{res}} = X_0^{-1}(\tilde{v}_{t}^{\mathrm{pred}} - \tilde{v}_{t})$ serves as an effective statistic for identifying topology changes,  as discussed next.
\begin{remark}
If a shortlist of plausible switches (or buses) is available, one can monitor only those buses and their neighboring buses to obtain $r_{t}^{\mathrm{res}}$ due to the sparsity of $X_0^{-1}$and $E$. In this work, however, we assume no prior knowledge and use full voltage measurements. 
\end{remark}

\subsubsection{Identification of the involved nodes}
\label{sec:node}
Because $r_t^{\mathrm{res}}$ lies in the range space of $E$, a topology reconfiguration produces a \emph{sparse} residual.  Once a topology change is detected, we take the detection timestamp $t$ returned by the algorithm in Section \ref{subsec:topo detect} as the start of the post-event window. We then stack the next $d$ residuals, using a short window to balance noise suppression and responsiveness, as follows,
$$R^{\mathrm{res}} = \bigl[\,r_t^{\mathrm{res}},\, r_{t+1}^{\mathrm{res}},\, \ldots,\, r_{t+d-1}^{\mathrm{res}}\,\bigr] \in \mathbb{R}^{N\times d}.$$
A bus $i$ is considered \emph{active} at time $t$ if its normalized score exceeds a small threshold~$\tau$:
\begin{equation}\label{eq:node_active}
    \frac{|r_t^{\mathrm{res}}(i)|}{\frac{1}{N}\sum_i^N|r_t^{\mathrm{res}}(i)|} > \tau.
\end{equation}
A practical choice for $\tau$ is to set it relative to the background noise level in the residuals. Finally, a bus is identified as \emph{involved} if it remains active for at least a fraction $\beta$ of the $p$ time steps, i.e.,
\begin{equation}\label{eq:node_id}
N_c = \Bigl\{\, i \;\Big|\; \tfrac{1}{p}\sum_{k=1}^{p}1\{i\text{ active at }t+k\} \ge \beta \Bigr\}.
\end{equation}
We use a high persistence threshold (e.g., $\beta = 0.8$) so that only buses showing sustained post-event activity are labeled as involved. This test provides fast and robust identification of nodes involved in the topology change. When $|N_c|< 2$, the identified buses cannot form a valid line, and the event is treated as a false alarm from the topology-change detector.

\subsubsection{Constrained Sparse Identification of Line Changes}
With the involved buses identified, we now determine the lines that changed status (removed or added). Let $\mathcal{E}_{\mathrm{cand}}$ be the set of $\binom{|N_c|}{2}$ candidate lines connected to the involved buses. We build the corresponding matrix $E_{\text{sub}}\in\mathbb{R}^{n\times |\mathcal{E}_{\mathrm{cand}}|}$. 
Since the number of changed lines is small relative to the candidate set, we exploit this sparsity by first applying LASSO regression with sign constraints in \eqref{eq:lasso-sign} to identify the changed lines. We then solve a regression restricted to the identified lines to estimate their parameters.

To encode “addition” vs. “deletion,” we partition candidates using the pre-event topology: if a candidate line already exists, it is placed in $\mathcal{E}_{\mathrm{del}}$; otherwise, it is placed in $\mathcal{E}_{\mathrm{add}}$. Denote by $\gamma\in\mathbb{R}^{|\mathcal{E}_{\mathrm{cand}}|}$ the diagonal of $\Gamma=\mathrm{diag}(\gamma)$, where each $\gamma_\ell$ corresponds to a candidate line $\ell$.
\begin{subequations}
\label{eq:lasso-sign}
\begin{align}
\intertext{\textbf{Sparse Identification of Changed Lines}}
\min_{\gamma\in\mathbb{R}^{|\mathcal{E}_{\mathrm{cand}}|}}\;&
\tfrac{1}{2}\,\bigl\| R^{res} - E_{\text{sub}}\,\mathrm{diag}(\gamma)\,E_{\text{sub}}^\top V \bigr\|_F^2
+\lambda \|\gamma\|_1, \label{eq:lasso-obj}\\[2mm]
\text{s.t.}\quad
& \gamma_\ell \ge 0\;\;(\ell\in\mathcal{E}_{\mathrm{add}}),\qquad
  \gamma_\ell \le 0\;\;(\ell\in\mathcal{E}_{\mathrm{del}}). \label{eq:lasso-signcons}
\end{align}
\end{subequations}
where $\lambda>0$ balances fit and sparsity, $V$ is the stacked voltage deviation, i.e., $V := \bigl[\,\tilde v_t,\; \tilde v_{t+1},\; \ldots,\; \tilde v_{t+d-1}\,\bigr] \in \mathbb{R}^{N \times d}$. It is a convex program and can be solved efficiently with off-the-shelf solvers. After solving \eqref{eq:lasso-sign}, we form the support
\begin{equation}\label{eq:support}
    S=\{\ell:\;|\hat{\gamma}_\ell|>\varepsilon_\gamma\},
\end{equation}
where $\varepsilon_\gamma$ is introduced to reject near-zero lines.  We then verify whether the proposed line reconfiguration yields a valid radial topology. If so, line parameter recovery is performed; otherwise, the candidate reconfiguration is rejected and the event is treated as a spurious alarm (e.g., caused by load fluctuations or measurement noise): the pre-event topology is retained, and the detection–identification procedure continues onto subsequent measurements.

To recover line parameters after selecting $S$, we remove the $\ell_1$ regularization {in \eqref{eq:lasso-sign}}, and estimate the line parameters in $\gamma$ by least square regression. Note that once the line parameters are estimated, the sensitivity matrix in \eqref{equ:dynamics:power_flow} can be corrected. 

\begin{algorithm}[t]
\caption{Topology Change Detection and Identification}
\label{alg:topo_detect_ident}
\begin{algorithmic}[1]
\State Inputs: history length $H$, window size $d$, thresholds $\tau,\beta$, LASSO weight $\lambda$
\State Initialize pre-event topology $X_0$ and empty error buffer $B$
\For{$t = 0,1,\dots$}
    \State Compute prediction error $e_t$ and append it to buffer $B$
    \State Compute the robust baseline $\kappa_t$ 
    \If{$\|e_t\|_2 > \kappa_t$}
    \State Flag a potential event
    \EndIf
    \If{a potential event was flagged at time $t$}
    \State Classify the event at time $t{+}1$
    \EndIf
    \If{a topology change is detected}
        \State Compute the residual $r_t^{\mathrm{res}}$ and append it to $R^{\mathrm{res}}$
        \If{$d$ residuals have been collected}
            \State Identify candidate nodes $N_c$ using \eqref{eq:node_id}
            \If{$|N_c| \ge 2$}
                \State Construct candidate edge set $\mathcal{E}_{\mathrm{cand}}$ from $N_c$
                \State Form the submatrix $E_{\text{sub}}$
                \State Solve the \eqref{eq:lasso-sign} to obtain line set $S$
                \If{$S$ defines a feasible topology}
                    \State Refine line parameters by least squares
                    \State Update the topology matrix $X_t$
                \EndIf
            \EndIf
            \State Reset topology-change flag
            \State Clear residual buffer $R^{\mathrm{res}}$
        \EndIf
    \EndIf
\EndFor
\end{algorithmic}
\end{algorithm}

We conclude this section with Algorithm \ref{alg:topo_detect_ident}. In the next section, we present our online voltage control algorithm, that is adaptable to topology changes. 
\section{Online Policy Optimization for Voltage Control}
\label{sec:online voltage control}
Recall the voltage control problem in \eqref{eq:volt-control-problem}, {with cost function $h_t$ given in \eqref{eq:cost}}. We now present an online policy optimization framework to update the voltage control policy $\pi(\cdot, \theta)$ using real-time measurement data. The framework is integrated with the topology change detection and identification algorithm introduced in the previous section to adapt to topology changes. Since the topology-change derivation was formulated for the full control set $\mathcal{P}$ including all buses, we use $X$ instead of $X_{\mathcal{P}}$ here for consistency.
\subsection{Online Policy Optimization Objective}
The goal of online policy optimization is to update parameter $\theta_t \!\to\! \theta_{t+1}$ so that performance of voltage control remains competitive with respect to time-varying system conditions. Following \cite{lin2023online}, online policy optimization aims to minimize the \emph{regret} with respect to the best fixed policy in hindsight, where the regret is defined as follows,
\begin{equation}
\label{eq:regret}
\mathrm{Reg}_T \;=\; \sum_{t=0}^{T-1}  h_t(\theta_t)\;-\;\min_{\theta}\sum_{t=0}^{T-1}\hat h_t(\theta).
\end{equation}
Here, $\hat h_t(\theta)$ is a surrogate cost that uses the same cost function, i.e., $\hat h_t(\theta):=h_t(\theta)$~\cite{lin2023online},  
which characterizes how good $\theta_t$ is at time $t$ if $\theta_t$ is applied from the start, without the impact of other historical policy parameters $\theta_{0:t-1}$. 

\subsection{Online Gradient-Based Policy Updates}
To minimize regret, we adopt the Memoryless Gradient-based Adaptive Policy Selection (M-GAPS) algorithm~\cite{pmlr-v247-lin24a}, which provides an efficient policy update rule based on an estimated system model. The resulting update takes the form:
\begin{equation}
    \theta_{t+1} \;\gets\; \theta_t - \eta\, G_{t+1}.
\end{equation}
where $G_{t+1}$ denotes the gradient of $h_{t+1}$ with respect to all past policy parameters $\theta_{0:t}$. Formally,
\begin{equation}
    G_{t+1} \coloneqq \sum_{\tau = 0}^{t} \frac{\partial h_{t+1}}{\partial \theta_\tau} =  \frac{\partial h_{t+1}}{\partial  q_{t+1}} \sum_{\tau = 0}^{t} \frac{\partial  q_{t+1}}{\partial \theta_\tau}.
\end{equation}
Intuitively, $G_{t+1}$ captures how the entire sequence of historical parameters $\theta_{0:t}$ contribute the current cost $h_{t+1}$. The dependency arises because the system state $q_{t+1}$, which denotes the reactive power setpoint at time step $t+1$, depends on all previous control actions $u_{0:t}$ and the initial state.

Using the cost function definition in \eqref{eq:cost} and the power model in \eqref{equ:dynamics:power_flow}, we can compute the gradient as
\begin{align}
\label{eq:gradient_computation}
    G_{t+1} = \left[ 2( v_{t+1} -  v_{nom})^\top {Q}_x\hat{X} + 2  q_{t+1}^\top {Q}_u \right]  y_{t+1}^q,
\end{align}
where 
$y_{t+1}^q \coloneqq \sum_{\tau = 0}^{t} \frac{\partial  q_{t+1}}{\partial \theta_\tau}$
is an auxiliary variable that accumulates the sensitivity of the reactive power state $q_{t+1}$ with respect to all past policy parameters. Note that $\hat{X}$ can be estimated using the proposed method described in Algorithm \ref{alg:topo_detect_ident}.
Equation \eqref{eq:gradient_computation} separates the gradient into two components: the \emph{instantaneous sensitivity} of the cost with respect to reactive power (the bracketed term), and the \emph{accumulated sensitivity} $y_{t+1}^q$, which captures how the current reactive power setpoint depends on all past policy parameters. The policy update naturally stops when either component becomes zero. 

\subsection{Recursive Computation for $y_{t+1}^q$}
The auxiliary state $ y_t^q$ admits a recursive form that enables efficient online computation. 
Since the reactive power evolves as $q_{t+1} = q_t +  u _t = q_t + \pi(v_t, \theta_t)$, substituting in $v_t = Xq_t + v_{env}$, we get 
\begin{equation}
    q_{t+1} = q_{t} + \pi({X} {q}_{t}+ {v}_{env}, \theta_t).
\end{equation}
By the chain rule, $y_{t+1}^q$ evolves as
\begin{equation}\label{eq:y_update}
     y_{t+1}^q = \underbrace{\frac{\partial q_{t+1}}{\partial q_t} \cdot  y_t^q}_{\text{Propagate past sensitivities}} + \underbrace{\frac{\partial q_{t+1}}{\partial \theta_t}}_{\text{Add current sensitivity}}.
\end{equation}
The Jacobian term $\frac{\partial q_{t+1}}{\partial q_t} =  {I} + \frac{\partial \pi}{\partial v_t}{X}$ propagates the influence of past parameters, while the term $\frac{\partial q_{t+1}}{\partial \theta_t} = \frac{\partial \pi}{\partial \theta_t}$ incorporates the effect of the current $\theta_t$. Both terms can be computed efficiently using automatic differentiation libraries (e.g., PyTorch). The recursive update in \eqref{eq:y_update}, together with the gradient expression in \eqref{eq:gradient_computation}, forms the core of the online policy optimization procedure summarized in Algorithm~\ref{alg:online_opt}. Note that the estimated $\hat{X}$, obtained from Algorithm \ref{alg:topo_detect_ident}, is required in both computations, and its accuracy impacts the overall cost.

\begin{algorithm}[t]
\caption{Online Policy Optimization with Topology Change Estimation}
\label{alg:online_opt}
\begin{algorithmic}[1]
\State \textbf{Parameters:} Learning rate $\eta$.
\State \textbf{Initialize:} Policy parameter $\theta_0$, auxiliary state $ y_0^q = 0$, initial state $v_0$ (initial voltage magnitude), $q_0$ (initial reactive power setpoint).
\State Retrieve initial sensitivity $\hat{X}$ from known topology.
\For {$t = 0, 1, \dots, T-1$}
    \If{a topology change is detected (Lemma \ref{def:detection})}
        \State 
        Re-estimate sensitivity matrix $\hat{X}$ by Algorithm \ref{alg:topo_detect_ident}
    \EndIf
    \State Calculate control based on the current policy parameter $\theta_t$: $ u _t = \pi(v_t, \theta_t)$.
    \State Update reactive power setpoint $q_{t+1} = q_t +  u _t$.
    \State Apply $q_{t+1}$ to the grid and measure next voltage $v_{t+1}$.
    \State Update auxiliary state $ y_{t+1}^q$ according to \eqref{eq:y_update}.
    \State Compute gradient $G_{t+1}$ according to \eqref{eq:gradient_computation}.
    \State Update policy parameters $\theta_{t+1} \gets \theta_t - \eta G_{t+1}$.
\EndFor
\end{algorithmic}
\end{algorithm}

\section{Experiments}\label{sec:experiments}
In this section, we evaluate the proposed topology change detection and identification framework and demonstrate its effectiveness in enhancing online voltage control. The code is available at \href{https://github.com/JieFeng-cse/Online-Optimization-for-NN-Voltage-Control}{https://github.com/JieFeng-cse/Online-Optimization-for-NN-Voltage-Control}.
\subsection{Experiment Setup} \label{sec:ep-setup}
\begin{figure}[tbp]
    \centering
    \begin{subfigure}[b]{0.3\linewidth}
        \centering
        \includegraphics[width=\linewidth]{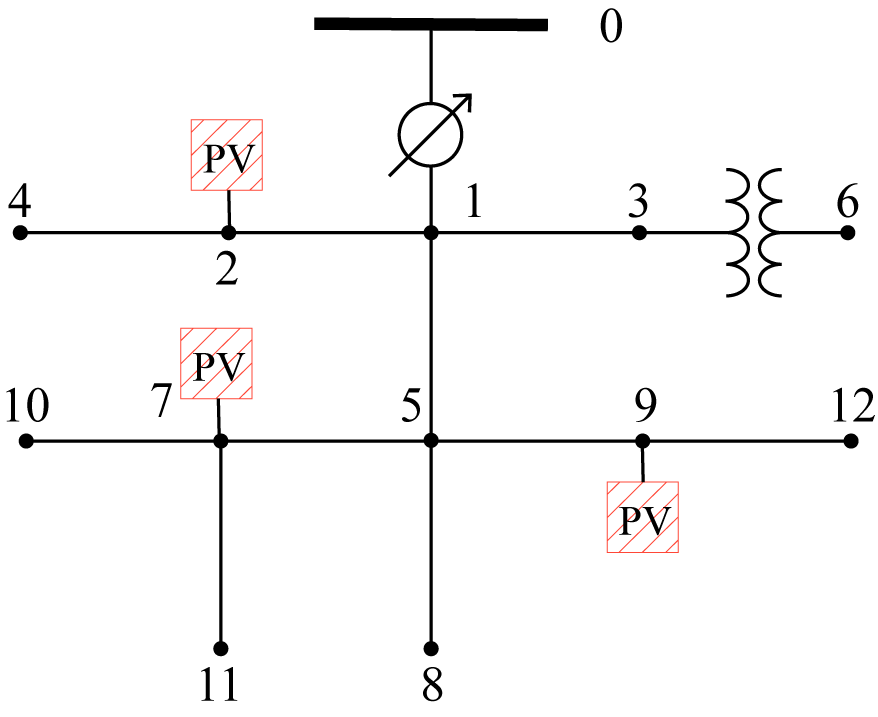}
        \label{fig:13bus}
    \end{subfigure}
    \hfill
    \begin{subfigure}[b]{0.65\linewidth}
        \centering
        \includegraphics[width=\linewidth]{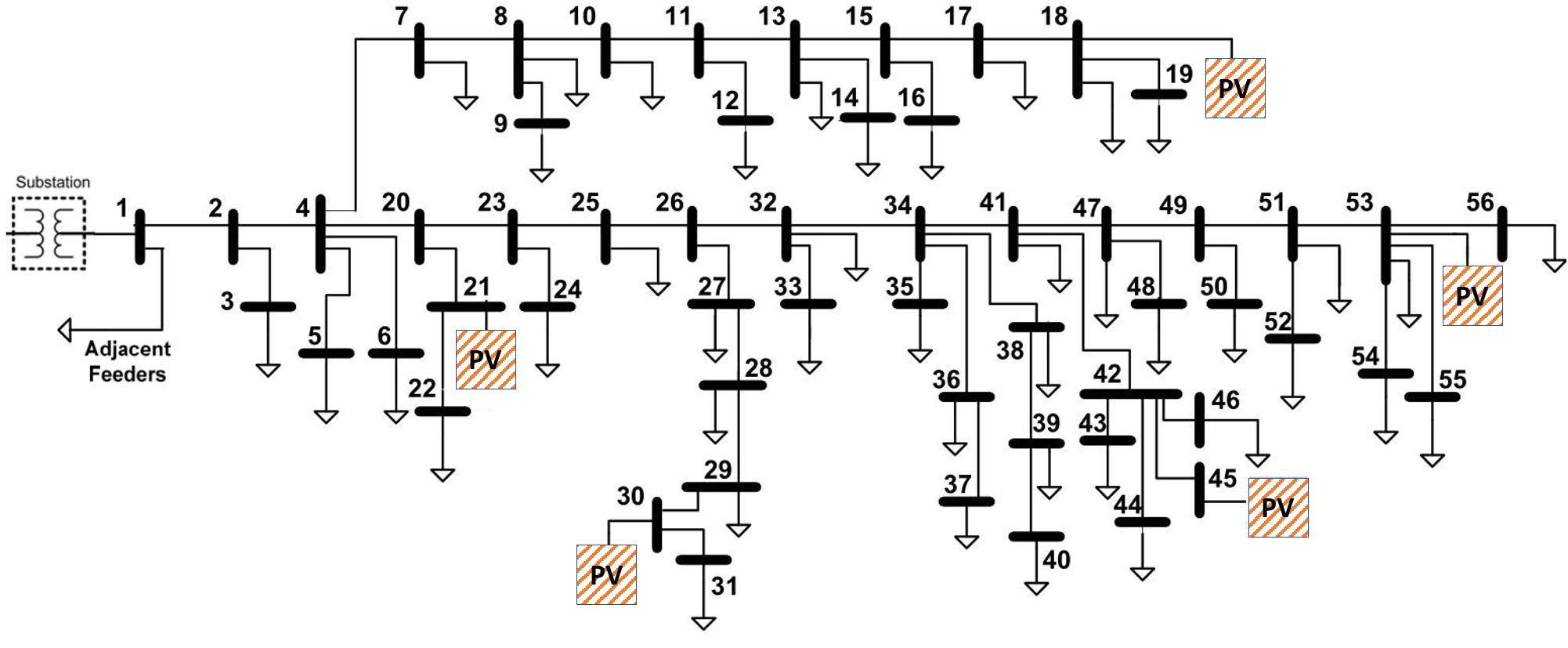}
        \label{fig:56bus}
    \end{subfigure}
    \caption{Diagram for IEEE 13-bus and SCE 56 bus systems.}
    \label{fig:sim_system}
\end{figure}

We evaluate the proposed framework on two radial test feeders: the IEEE 13-bus system (single-phase version of~\cite{8063903}) and the Southern California Edison (SCE) 56-bus network~\cite{6345736}. The system diagram is shown in Fig.~\ref{fig:sim_system}. While the controller and theory are derived based on the linearized model, all experiments are conducted using the nonlinear power flow model implemented in Pandapower~\cite{pandapower}.

Two disturbance scenarios are simulated to induce realistic voltage deviations:  
(1) \emph{High Voltage}, daytime conditions with high photovoltaic (PV) generation, and  
(2) \emph{Low Voltage}, heavy-load conditions with no PV output.  
In each case, one scenario is selected uniformly at random, and active power injections are randomly perturbed to create voltage deviations of approximately 5–15\% from nominal values. The configuration of each system is described below.

\paragraph{IEEE 13-Bus System}
The nominal voltage is 4.16~kV with an acceptable voltage range $[3.952, 4.368]$~kV (i.e., $\pm5\%$ of nominal). Three PV units are installed at buses 2, 7, and 9~\cite{10336939} that can flexibly adjust reactive power injection for voltage control. 
To assess the framework's adaptability, we evaluate eight topology reconfiguration scenarios involving up to four line switching actions, presented in Table \ref{tab:ieee13_reconfig}.
\begin{table}[h]
\centering
\caption{Topology Reconfiguration Scenarios for the IEEE 13-Bus System}
\label{tab:ieee13_reconfig}
\renewcommand{\arraystretch}{1.2}
\begin{tabular}{c l}
\toprule
\textbf{Scenario} & \textbf{Line Switching Actions (Disconnect $\rightarrow$ Connect)} \\
\midrule
1 & $(1,5) \rightarrow (3,9)$ \\
2 & $(1,5) \rightarrow (3,12)$ \\
3 & $(1,5) \rightarrow (2,7)$ \\
4 & $(1,5) \rightarrow (1,9)$ \\
5 & $(5,7) \rightarrow (2,7)$ \\
6 & $(5,7) \rightarrow (4,7)$ \\
7 & $(1,5),(5,7) \rightarrow (2,7),(3,9)$ \\
8 & $(5,9),(5,7) \rightarrow (3,9),(4,7)$ \\
\bottomrule
\end{tabular}
\end{table}

\paragraph{SCE 56-Bus System}
The nominal voltage is 12~kV, with an acceptable range of $[11.4, 12.6]$~kV (i.e., $\pm5\%$ of nominal). Five controllable PVs are placed at buses 18, 21, 30, 45, and 53~\cite{shi2021stability}. The considered topology change scenarios are presented in Table \ref{tab:sce56_reconfig}.
\begin{table}[t]
\centering
\caption{Topology Reconfiguration Scenarios for the SCE 56-Bus System}
\label{tab:sce56_reconfig}
\renewcommand{\arraystretch}{1.2}
\begin{tabular}{c l}
\toprule
\textbf{Scenario} & \textbf{Line Switching Actions (Disconnect $\rightarrow$ Connect)} \\
\midrule
1 & $(34,41) \rightarrow (2,41)$ \\
2 & $(47,49) \rightarrow (18,49)$ \\
3 & $(13,15) \rightarrow (15,34)$ \\
4 & $(41,42) \rightarrow (46,47)$ \\
5 & $(47,49) \rightarrow (11,49)$ \\
6 & $(34,41),(47,49) \rightarrow (2,41),(10,49)$ \\
7 & $(34,41),(23,15) \rightarrow (2,41),(15,34)$ \\
8 & $(41,42),(32,34) \rightarrow (20,34),(42,48)$ \\
\bottomrule
\end{tabular}
\end{table}

During each simulation, the system starts from its original topology with a randomly generated initial voltage deviation, under the assumption that the original topology is known. The topology change detection algorithm is activated at the beginning of each run, and online optimization begins using the true original topology. A random event from the predefined list described above is sampled to induce a topology change after the start. The sensitivity estimation algorithms have no prior knowledge of possible reconfigurations. Reactive power capacity limits are ignored in this study but can be incorporated in practice as a soft penalty in the RL reward function or as a hard constraint by projection.

\subsection{Neural Network Policy Design and Pre-training}
Following~\cite{10336939}, we employ a decentralized monotone neural network policy at each controllable bus, ensuring closed-loop voltage stability. Each bus $i$ applies a local control policy $\pi_i(v_{i,t},\theta_{i,t})$ parameterized by a monotone neural network, that is parameterized by $\theta_{i,t}=(\theta_{i,t}^1,\theta_{i,t}^2)$. The vector $\theta_{i,t}^1 \in \mathbb{R}^{d_\theta}$ contains the weights of the monotone neural network, while $\theta_{i,t}^2 \in \mathbb{R}$ is a scalar parameter that adjusts the local voltage setpoint. The control law is given by
\begin{subequations}\label{eq:controller}
    \begin{align}
    u_{i,t } &= -[\xi_{\theta_{i,t}^1}^{+}(v_{i,t}-b_{\theta_{i,t}^2}) + \xi_{\theta_{i,t}^1}^{-}(v_{i,t}-b_{\theta_{i,t}^2})],\\
    b_{\theta_{i,t}^2}&=v_{\text{nom},i} + 0.05 \cdot v_{\text{nom},i} \cdot \tanh(\theta_{i,t}^2). \label{eq:setpoint}
    \end{align}
\end{subequations}
Here, $\xi^+$ and $\xi^-$ are monotonically increasing neural-network components constructed as in~\cite{10336939}. The two components are activated above and below the learned setpoint $b_{\theta_{i,t}^2}$, respectively.
The negative sign ensures an overall monotonically decreasing control law, a sufficient condition for voltage stability \cite{10336939}. The learnable setpoint $b_{\theta_{i,t}^2}$ is constrained by \eqref{eq:setpoint} to always lie within the safe voltage range.

The parameters $\theta_i$ are pre-trained using the Stable Deep Deterministic Policy Gradient (Stable-DDPG) algorithm~\cite{10336939} under the known initial topology. This produces a provably stable controller that serves as the initial policy for online adaptation. During online operation, these parameters are updated over time, becoming $\theta_{i,t}$ at each time step $t$. For all experiments below, we run the complete framework using online policy optimization with a learning rate of $\eta = 0.1$ (unless otherwise specified) and pre-trained neural controllers. We choose a relatively large learning rate to accelerate adaptation to changes in the system dynamics.

\subsection{Topology Change Detection and Identification}
We first evaluate the performance of the proposed topology change detection and identification framework in Algorithm \ref{alg:topo_detect_ident}. For topology change detection, the error buffer size is set to $H=5$. For topology identification, the data collection window is set to $d=10$ for the IEEE 13-bus system and $d=15$ for the SCE 56-bus system.
Once the window is filled, the residual matrix $R^{\text{res}}\in\mathbb{R}^{N\times d}$ is used to (i) identify candidate nodes via~\eqref{eq:node_active} and \eqref{eq:node_id} with threshold $\tau=10^{-2}$, (ii) select line candidates through the sparse LASSO formulation~\eqref{eq:lasso-sign} and set $\epsilon_\gamma=1.0$ in \eqref{eq:support} , and (iii) re-estimate the corresponding line parameters by regression to confirm the final topology.

Each trajectory consists of 1000 steps and begins from the original topology with an initial voltage deviation, as described in Section~\ref{sec:ep-setup}. A random topology reconfiguration is introduced at step 50, and a small random load change occurs every 200 steps. Over 200 simulated trajectories, the proposed method achieves a 100\% event-detection success rate\footnote{An event is considered successfully detected if a topology change is flagged following a true reconfiguration. False positives may still occur but are later ruled out during topology identification—for example, cases with fewer than two active nodes or violations of the radiality (tree) condition.}. The accuracy of the node inclusion step slightly decreases due to nonlinear effects. 
Based on the identified nodes, the method maintains high accuracy in both candidate line selection and final topology identification, as summarized in Table~\ref{tab:topology_performance}. 

Despite the high overall accuracy, a few failures occur in the IEEE 13-bus system. These typically arise when several nodes experience unusually large voltage deviations, causing a few rows of the residual matrix $R^{\text{res}}$ to dominate the regression. Because each row corresponds to a node’s voltage residual, such imbalance skews the optimization and leads to incorrect line selection.
For example, in one failure case, line $(1,5)$ is disconnected while line $(1,9)$ is added. The resulting voltage residuals at buses 5 and 9 are significantly larger than those at other buses, which biases the regression toward incorrectly identifying a spurious line $(5,9)$.
This issue is exacerbated by the quadratic loss term in \eqref{eq:lasso-obj}; replacing it with an $\ell_1$-norm regression error can mitigate the effect but may slightly reduce accuracy in other scenarios. Moreover, since the method is derived from a linearized model, large voltage deviations can amplify nonlinear effects, potentially leading to increased estimation error or even incorrect topology identification.

\begin{table}[t] 
\centering 
\begin{tabular}{lcccc} 
\toprule Environments & Event & Node & Line & Final \\ & detect & inclusion & inclusion. & identification \\ 
\midrule IEEE 13-bus & 1.00 & 0.98 & 0.96 & 0.90 \\ 
SCE 56-bus & 1.00 & 1.00 & 1.00 & 0.98 \\ 
\bottomrule 
\end{tabular} 
\caption[Success rates of detection and identification]{Success rates of event detection, node and line candidate inclusion, and final identification. The ``Final identification'' metric counts a trial as correct only when the inferred topology change exactly matches the true change.} 
\label{tab:topology_performance} 
\end{table}

\subsection{Control Performance}
Table~\ref{tab:comparison} compares control performance across different sensitivity updating strategies. The cumulative cost is defined in~\eqref{eq:question,cost}. {To compare our approach based on topology change detection and identification, we estimate control cost for RLS and OLS based sensitivity estimated, mentioned in Section \ref{sec:sub_sensitivity}.} For RLS and OLS, we adopt the standard RLS formulation and augment OLS with ridge regularization~\cite{7997787} to mitigate collinearity in operating data. RLS is initialized using the last known sensitivity before a topology change. Details are included in Appendix \ref{appendix:sensitivity}.
The estimation time is the time step when the voltage deviation prediction error falls below $10^{-4}$; otherwise, the time is capped at 1000. 
The sensitivity error (Err. (norm)) denotes the matrix norm between the estimated and true sensitivities, i.e., $\lVert X_{\mathcal{P},t} - \hat{X}_{\mathcal{P},t}\rVert_2$. 
Online policy optimization is suspended whenever a topology-change event is detected, until a new sensitivity estimate becomes available.

Without online adaptation, the pre-trained policy performs poorly since it is optimized for an average operating condition under a fixed topology. As a result, it becomes suboptimal for specific load profiles and degrades significantly after topology changes. Incorporating online adaptation (via OLS, RLS, or the proposed method) enables the controller to adjust to specific load conditions and network topology, thereby significantly improving the overall performance.

Given that topology changes typically affect only a small subset of lines compared to the total number of lines in the network, our identification algorithm focuses only on the lines directly involved in the change.
This localized design allows the update law to scale with the number of changed lines rather than the full network size, substantially reducing data requirements. In contrast, regression-based estimators such as OLS and RLS scale with system size, resulting in longer estimation time and larger errors. OLS further suffers from ill-conditioning even with ridge regularization, while RLS improves convergence speed but remains slower to adapt. By accurately identifying topology changes (success rate $>90\%$), our method achieves the lowest overall control cost, reducing it by 5\% in the IEEE 13-bus system and 25\% in the SCE 56-bus system. The average estimation time is longer in the IEEE 13-bus due to the $10\%$ failed cases, where the entire trajectory is used for sensitivity estimation. If OLS was deployed after those failed cases, the average estimation time would be $26.41$.
\begin{table}[t]
\centering
\begin{tabular}{lcccc}
\toprule
Environments & Method & Cost & Est. time & Err. (norm) \\
\midrule
IEEE 13-bus & Fixed  & 38.54 & - & - \\
            & OLS  & 13.03 & 72.38 & 0.08 \\
            & RLS  & 13.20 &  55.16 & 0.08 \\
            & \textbf{Ours} & \textbf{12.46} & \textbf{101.94 (26.41) }  & \textbf{0.016} \\
\midrule
SCE 56-bus  & Fixed  & 266.81 & - & - \\
            & OLS  & 114.44 & 255.56 & 0.11 \\
            & RLS  & 115.73 & 247.78 & 0.10 \\
            & \textbf{Ours} & \textbf{85.78} & \textbf{57.42 }  & \textbf{0.006} \\
\bottomrule
\end{tabular}
\caption{Performance comparison across different updating schemes, averaged over 200 trajectories of 1000 steps each. For the IEEE 13-bus system, if OLS is deployed after the failure of ours, the average estimation time would be $26.41$.}
\label{tab:comparison}
\end{table}

\subsection{Voltage Control with Real World Loads}
\begin{figure}
    \centering
    \includegraphics[width=0.95\linewidth]{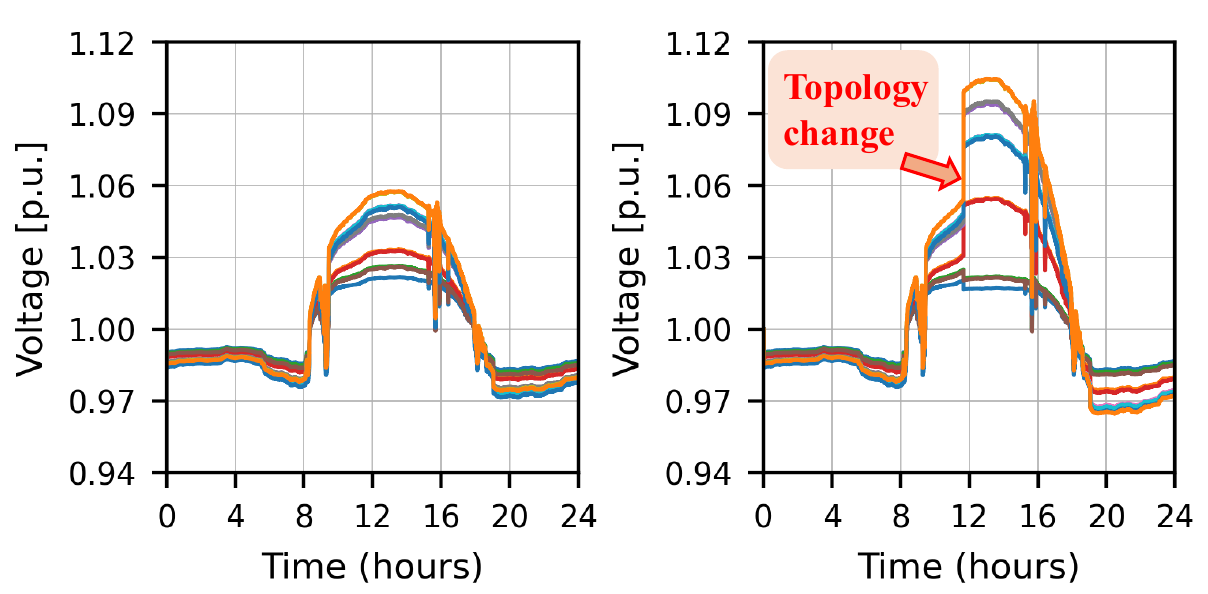}
    \caption{IEEE 13-bus voltage trajectories without control under real-world loads. 
Left: 24-hour voltage profile with no topology changes.
Right: voltage profile when a topology change occurs at 11{:}40~a.m., disconnecting line (1,5) and connecting line (2,7). 
}
    \label{fig:real_load_no_ctrl}
\end{figure}
\begin{figure}
    \centering
    \includegraphics[width=0.95\linewidth]{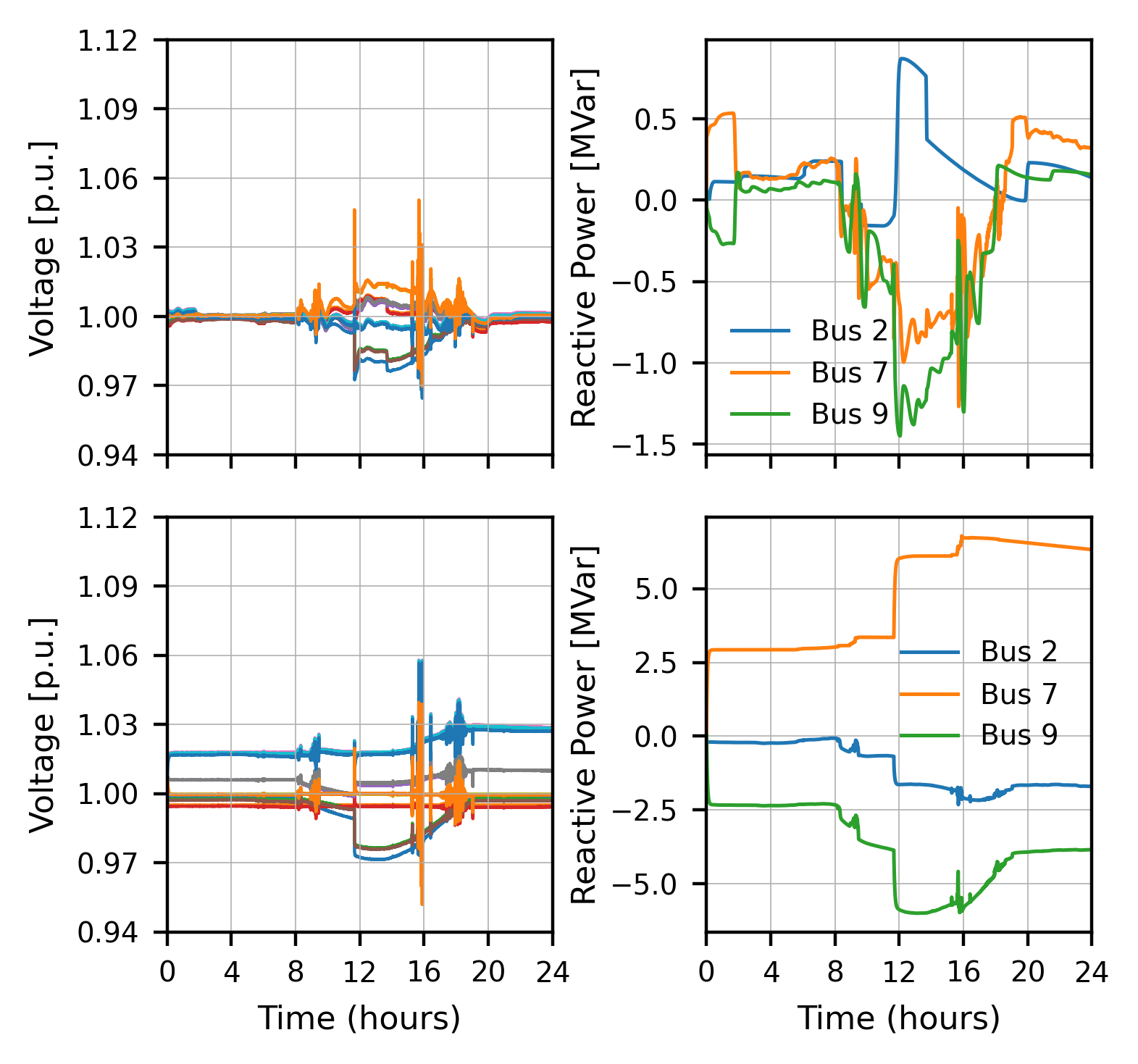}
    \caption{IEEE 13-bus voltage trajectories with control, under the same real-world load profile and topology change as Fig.~\ref{fig:real_load_no_ctrl}. 
Top row: voltage (left) and reactive power injections (right) under the proposed adaptive control method. 
Bottom row: voltage (left) and reactive power injections (right) using the pre-trained control policy without online adaptation. 
}
    \label{fig:real_load_online_control}
\end{figure}
Finally, we evaluate the proposed method on the IEEE 13-bus system using real-world load and PV generation profiles over a full day~\cite{10336939}. Each load and generation data point is repeated six times to reflect that the controller operates six times faster than the load variations.
Figure~\ref{fig:real_load_no_ctrl} shows the voltage trajectory without control. The topology change is introduced at 11{:}40~a.m. abruptly alters power-flow paths, resulting in a sharp voltage disturbance. As shown in the first row of Figure~\ref{fig:real_load_online_control}, our method accurately detects and identifies the topology change while maintaining stable voltage regulation.
In contrast, the fixed pre-trained policy (second row) eventually restores voltages to acceptable levels but requires excessive reactive power injection, incurring nearly 30$\times$ higher cumulative control cost compared to ours. 

\section{Conclusion}\label{sec:conclusion}
This paper presented an online policy optimization framework for voltage control in distribution grids under unknown topology changes. The proposed method leverages data-driven estimation of voltage–reactive power sensitivities to enable efficient policy adaptation under changing system conditions. By exploiting the sparsity of topology change events and enforcing radial network topology constraints, the estimation algorithm can accurately detect topology change events and identify the changed lines.
Simulations on the IEEE 13-bus and SCE 56-bus systems show that the proposed approach achieves over 90\% line identification accuracy with short observation windows and substantially improves voltage regulation performance compared with both non-adaptive and adaptive baselines using regression-based sensitivity estimation. While the current framework is derived from a linearized model and validated on nonlinear systems, future work will focus on directly analyzing and designing from nonlinear dynamics to better capture real system behavior. 
We also plan to investigate its robustness under communication delays, measurement noise, and potential cyber attacks.

\section*{Acknowledgment}
The work of Y. Shi and J. Feng are supported by the National Science Foundation under Grant ECCS-2442689. The work of J. Feng is also supported by the UC-National Laboratory In Residence Graduate Fellowship L24GF7923. The work of D. Deka is supported by the MIT Energy Initiative and Los Alamos National Laboratory LDRD program as part of the Artimis project. Any opinions, findings and conclusions in this paper are of the authors and do not necessarily reflect the views of the funding agencies. 

\bibliographystyle{ieeetr}
\bibliography{bib/alias,bib/JC,bib/Main-add}

\appendices
\section{Regression-Based Sensitivity Estimation}\label{appendix:sensitivity}

This appendix summarizes the formulations of the ordinary least squares (OLS) and recursive least squares (RLS) methods used for sensitivity estimation in data-driven voltage control.

\subsection{Ordinary Least Squares (OLS)}

From the linearized power flow relation
\begin{equation}
    \tilde{v}_{t+1} = X_{\mathcal{P}} u_t, \label{eq:app_vdiff}
\end{equation}
where $\tilde{v}_{t+1} = v_{t+1}-v_t$ denotes the voltage change and $u_t\in\mathbb{R}^M$ is the control input at time $t$, the goal is to estimate the sensitivity matrix $X_{\mathcal{P}}\in\mathbb{R}^{N\times M}$.
Stacking $T$ samples of inputs and outputs yields
\begin{equation}
    \tilde{V} = X_{\mathcal{P}} \Delta, \label{eq:app_ols_matrix}
\end{equation}
where $\tilde{V}=[\tilde{v}_1, \tilde{v}_2, \ldots, \tilde{v}_{T}] \in \mathbb{R}^{N\times T}$ and $\Delta=[u_0, u_1, \ldots, u_{T-1}] \in \mathbb{R}^{M\times T}$.
The OLS estimate minimizes the squared residual error:
\begin{equation}
    \min_{X_{\mathcal{P}}} \; \sum_{t=0}^{T-1} \| \tilde{v}_{t+1} - X_{\mathcal{P}} u_t \|_2^2.
\end{equation}
When $\Delta$ has full row rank, the analytical solution is
\begin{equation}
    \hat{X}_{\mathcal{P}} = \tilde{V}\Delta^{+} = \tilde{V}\Delta^\top(\Delta\Delta^\top)^{-1},
\end{equation}
where $(\cdot)^{+}$ denotes the Moore–Penrose pseudoinverse.
OLS provides an unbiased estimate under a fixed topology for the linear power flow model.

\subsection{Recursive Least Squares (RLS)}

To enable online adaptation, RLS updates the sensitivity estimate as new data come in.
For notational convenience, vectorize $\hat{X}_{\mathcal{P}}$ into $\omega_t = \mathrm{vec}(\hat{X}_{\mathcal{P}}) \in \mathbb{R}^{NM}$, and define the regression matrix $\Phi_t = I_N \otimes u_t^\top \in \mathbb{R}^{N\times (NM)}$, where $\otimes$ denotes the Kronecker product. The RLS objective is
\begin{equation}
    J_{\text{RLS}}(t) = \sum_{i=1}^{t}\lambda^{t-i} 
        \| \tilde{v}_i - \Phi_{i-1} \omega_t \|_2^2,
\end{equation}
where $0<\lambda\le1$ is a forgetting factor that downweights older samples.
The recursive update equations are:
\begin{subequations}\label{eq:app_rls}
\begin{align}
    p_0 &= \alpha I, \qquad \alpha \gg 1,\\
    K_t &= \frac{p_{t-1}\Phi_{t-1}^\top}{\lambda + \Phi_{t-1} p_{t-1}\Phi_{t-1}^\top},\\
    \omega_t &= \omega_{t-1} + K_t (\tilde{v}_t - \Phi_{t-1} \omega_{t-1}),\\
    p_t &= \frac{1}{\lambda}\left(p_{t-1} - K_t \Phi_{t-1}^\top p_{t-1}\right),
\end{align}
\end{subequations}
where $p_t\in\mathbb{R}^{(NM)\times (NM)}$ is the error covariance matrix and $K_t$ is the Kalman gain vector.
The initialization $p_0=\alpha I$ reflects high initial uncertainty.
For large networks, updating $p_t$ can be expensive. 
A common simplification is to run $N$ parallel scalar-output RLS estimators, each updating one row of $X_{\mathcal{P}}$, while sharing a covariance matrix $p_t\in\mathbb{R}^{M\times M}$ and gain $K_t\in\mathbb{R}^{M\times N}$ under the assumption that voltage measurement noise is independent across buses \cite{haykin2002adaptive}.

\end{document}